\documentclass[12pt,psamsfonts]{amsart}

\usepackage{fullpage}
\usepackage[parfill]{parskip}  
\usepackage[foot]{amsaddr}  

\usepackage{color}

\usepackage{enumitem}
\usepackage{caption}
\usepackage{subcaption}
\usepackage{amsmath,amssymb,amsthm}    
\usepackage{graphicx}

\usepackage{hyperref}   

\usepackage{amsfonts}
\usepackage{fullpage}
\newtheorem{theorem}{Theorem}[section]

\newtheorem{lemma}[theorem]{Lemma}
\newtheorem{proposition}[theorem]{Proposition}

\theoremstyle{definition}

\newtheorem{remark}[theorem]{Remark}

\usepackage{xy}
\xyoption{matrix}
\xyoption{frame}
\xyoption{arrow}
\xyoption{arc}

\usepackage{ifpdf}
\ifpdf
\else
\PackageWarningNoLine{Qcircuit}{Qcircuit is loading in Postscript mode.  The Xy-pic options ps and dvips will be loaded.  If you wish to use other Postscript drivers for Xy-pic, you must modify the code in Qcircuit.tex}
\xyoption{ps}
\xyoption{dvips}
\fi

\entrymodifiers={!C\entrybox}

\begin{document}
\title[]{ Estimating the Ratio of Two Functions in a Nonparametric Regression Model } 

\author[J. Markovic]{Jelena Markovic}
\thanks{Research supported by MIT's Undergraduate Research Opportunities Program (UROP)}

\email{markovic@mit.edu}
\author[L. Wang]{Lie Wang}
\email{liewang@math.mit.edu}
\address{Massachusetts Institute of Technology, Cambridge, MA, USA. }

\begin{abstract}

Due to measurement noise, a common problem in in various fields is how to estimate the ratio of two functions. We consider this problem of estimating the ratio of two functions in a nonparametric regression model. Assuming the noise is normally distributed, this is equivalent to estimating the ratio of the means of two normally distributed random variables. We identified a consistent estimator that gives the mean squared loss of order $O(1/n)$ ($n$ is the sample size) when conditioned on a highly probable event. We also present our result applied to both the real data from EAPS and on simulated data, confirming our theoretical results.

\end{abstract}

\maketitle


\section{Introduction}
\label{sec:intro}


Due to measurement noise, a common problem in physicochemical experiments is estimating the ratio of two signals (\cite{dodson}, \cite{ludwig}). The double-interpolation method of Dodson (\cite{dodson}) does not assume any measurement noise, making the computed error larger. This estimation problem also appeared in the Department of Earth, Atmospheric and Planetary Sciences (EAPS) at MIT, where scientists needed to estimate the ratio of a rock's gradient. In this work, we use a nonparametric regression model to estimate the ratio of two functions. In this setting, we identified a consistent estimator that has a small mean squared loss. We also present our result applied to both the real data from EAPS and on simulated data, confirming our theoretical results.

A nonparametric regression model is often used for various estimations, as in \cite{muller}, \cite{hall} and \cite{wang}. 
Similarly, our model is a nonparametric regression model in which we observe 
\begin{equation} \label{eq:set1}
x_{i}=rf(t_{i})+z_{i}\;\;\; (i=1,2,\ldots,n)
\end{equation}
and
\begin{equation}
\label{eq:set2}
y_{i}=f(t_{i})+w_{i}\;\;\;(i=1,2,\ldots,n), 
\end{equation}
$t_{i}\in\mathbb{R}$, where $z_{i}\sim\mathcal{N}(0,\sigma_{1}^{2})$ and $w_{i}\sim\mathcal{N}(0,\sigma_{2}^{2})$ are independent random variables and $f(\cdot)$ is an unknown function. 
The main parameters of the model and their essential role in the problem are as follows:
\begin{itemize}
\item $n$, which is the number of samples and should be assumed to be a moderately large number,
\item $(x_{1},x_{2},\ldots,x_{n})$ and $(y_{1},y_{2},\ldots,y_{n})$ are the observations,
\item $\sigma_{1}^2$ and $\sigma_{2}^2$ are the variances of the noise in measuring $rf$ and $f$, respectively. We did not know these variances a priori, but as we shall see in Section \ref{sec:variance}, these variances can be estimated with ease, hence in analyzing our proposed estimator, we work under the assumption that they are given to us.
\end{itemize}
We were not interested in the unknown function $f$, rather in the ratio $r$, which we wanted to estimate based on the observations $(x_{1},\ldots,x_{n})$ and $(y_{1},\ldots,y_{n})$. The estimation accuracy of $r$ is measured by the mean squared loss, 
\begin{equation*} \label{eq:loss}
R(\hat{r},r)=\mathbb{E}\left[(\hat{r}-r)^2\right],
\end{equation*}
where $\hat{r}$ denotes the estimated value of the true ratio $r$.
Let us denote with $X=(x_{1},\ldots,x_{n})$ and $Y=(y_{1},\ldots, y_{n})$ the observed data, and with $\mu_{1}=\left(rf(t_{1}),\ldots,rf(t_{n})\right)$ and $\mu_{2}=\left(f(t_{1}),\ldots,f(t_{n})\right)$ the unknown means. The goal is to estimate $r$ based on $X\sim\mathcal{N}(\mu_{1},I_{n}\sigma_{1}^2)$ and $Y\sim\mathcal{N}(\mu_{2},I_{n}\sigma_{2}^2)$, where $\mu_{1}=r\mu_{2}$ and the vectors $\mu_{1}, \mu_{2}\in\mathbb{R}^{n}$ are unknown.

In order to estimate the ratio $r$ of $\mu_{1}$ and $\mu_{2}$, note that $r$ is the least square solution, 
\begin{equation*}
r=\underset{c\in\mathbb{R}}{\text{arg}\;\text{min}}\|\mu_{1}-c\mu_{2}\|^2.
\end{equation*}
A naive approach is to replace $\mu_{1}$ with $X$ and $\mu_{2}$ with $Y$ in the above equation and estimate $r$ by 
\begin{equation*}
\tilde{r}=\underset{c\in\mathbb{R}}{\text{arg}\;\text{min}}\|X-cY\|^2=\frac{\langle X,Y\rangle}{\|Y\|^2}.
\end{equation*}
Note that $\mathbb{E}\left[\langle X,Y\rangle\right]=r\|\mu_{2}\|^2$ and $\mathbb{E}\left[\|Y\|^2\right]=\|\mu_{2}\|^2+n\sigma_{2}^2$. Thus, $\tilde{r}$ is not a good estimate for $r$ when $n$ is large. A better approach is to estimate $\|\mu_{2}\|^2$ by $\|Y\|^2-n\sigma_{2}^2$, which leads to the following estimator:
\begin{equation*}
\hat{r}=\frac{\langle X,Y\rangle}{\|Y\|^2-n\sigma_{2}^2}.
\end{equation*}
We need to analyze the loss, $R(\hat{r},r)=\mathbb{E}[(\hat{r}-r)^2]$, of this estimator. Let us denote with $D(Y)=\|Y\|^2-n\sigma_{2}^2$ the denominator of the proposed estimator. Since $D(Y)$ can take values arbitrarily close to zero, the mean squared loss of the proposed estimator is infinite as shown in Proposition \ref{prop:infinite_loss}. However, when we impose the condition $D(Y)>\alpha\mathbb{E}[D(Y)]$ for some constant $\alpha>0$, the loss is small.

\begin{proposition}
\label{prop:infinite_loss}
For $\hat{r}=\frac{\langle X,Y\rangle}{\|Y\|^2-n\sigma_{2}^2}$ its mean squared loss, $\mathbb{E}\left[(\hat{r}-r)^2\right]$, is infinite.
\end{proposition}

If we assume $\|\mu_{2}\|^2=\Omega(n)$, from Lemma \ref{lemma:D}, we conclude that $\mathbb{P}\left(|D(Y)|\leq\alpha\mathbb{E}[D(Y)]\right)=O(1/n)$ for any constant $\alpha\in(0,1)$. Thus, the denominator of the proposed estimator is bounded away from 0 with high probability. We show in Proposition \ref{prop:msl} that the mean squared loss conditioned on the event that the denominator is bounded away from 0 is in $O(1/n)$. In formal terms:

\begin{lemma} \label{lemma:D} 
For any constant $\alpha\in(0,1)$, the following holds: 
\begin{equation*}
\mathbb{P}\left(|D(Y)|\leq \alpha\mathbb{E}[D(Y)]\right)=O\left(\frac{\|\mu_{2}\|^2+n}{\|\mu_{2}\|^4}\right).
\end{equation*}
Furthermore, if we assume $\|\mu_{2}\|=\Omega(n)$, we get
\begin{equation}
\mathbb{P}\left(|D(Y)|\leq \alpha\mathbb{E}[D(Y)]\right)= O\left(\frac{1}{n}\right).
\end{equation}
\end{lemma}

\begin{proposition}
\label{prop:msl}
Assuming $\|\mu_{2}\|^2=\Omega(n)$, 
\begin{equation*}
\mathbb{E}\left[(\hat{r}-r)^2 \;\Big|\; |D(Y)|>\alpha\mathbb{E}[D(Y)]\right]=O\left(\frac{1}{n}\right)
\end{equation*}
holds for any constant $\alpha\in(0,1)$.
\end{proposition}

\begin{remark}\label{remark:condition}
One should notice that the condition $\|\mu_{2}\|^2=\Omega(n)$ is not too restrictive since it holds given that $f$ is bounded away from zero, i.e., there is a constant $c>0$ such that $f(x)\geq c$ for all $x\in[0,1]$; this would imply $\|\mu_{2}\|^2=\sum_{i=1}^{n}f(t_{i})^2\geq nc^2$. 
\end{remark}

\vspace{7 mm}

Furthermore, in the following proposition, we prove that the proposed estimator, $\hat{r}=\frac{\langle X,Y\rangle}{\|Y\|^2-n\sigma_{2}^2}$, is consistent.

\begin{proposition}
\label{prop:consistency}
Assuming $\|\mu_{2}\|^2=\Omega(n)$, we have
\begin{equation*}
\mathbb{P}\left((\hat{r}-r)^2\leq \frac{c\log(n)}{n}\right)\rightarrow 1,\;\;\text{ as } n\rightarrow\infty,
\end{equation*}
for $\hat{r}=\frac{\langle X,Y\rangle}{\|Y\|^2-n\sigma_{2}^2}$ and any constant $c>0$.
\end{proposition}

{\bf Techniques used in the proofs.} All the proofs of the results above are given in Section \ref{sec:proofs}. Here, we briefly elaborate on the techniques used. For proving {\bf Proposition \ref{prop:infinite_loss}}, we first used a conditional expectation to get a lower bound on the loss $\mathbb{E}_{X,Y}[(\hat{r}-r)^2]$ only in terms of $\mathbb{E}_{Y}$. Essentially, this step allowed us to get rid of the random variable $X$. Second, we used the properties of noncentral chi-squared distribution $\chi'^{2}_{n}$ to prove that this lower bound is actually infinite, making our loss infinite as well. {\bf Lemma \ref{lemma:D}} can be proved by conveniently using Chebyshev inequality. Afterwords, we used this lemma to prove Proposition \ref{prop:msl}. 
Finally, {\bf Proposition \ref{prop:consistency}} uses the law of total probability to analyze the probability of our estimator $\hat{r}$ being close to the true value $r$ in each of the following events: the denominator $D$ is far from its mean, $|D-\mathbb{E}[D]|>\mathbb{E}[D]/2$; and the denominator $D$ is close to its mean, $|D-\mathbb{E}[D]|<\mathbb{E}[D]/2$. Analyzing the latter is nontrivial, requiring several inequalities conveniently put together.

{\bf Outline of the paper.} In the following sections, we provide some observations that are useful in the applications. In Section \ref{sec:special}, we present a special case when $\|\mu_{2}\|^2=\Theta (n)$, for which the highly probable event we are conditioning on can be checked from data, ensuring the mean squared loss is small. Furthermore, it is also useful to have a method for estimating the noise variances. One should notice that for computing our estimator $\hat{r}$ it is enough to know $\sigma_{2}$, the variance of noise when observing $f$. In Section \ref{sec:variance}, we describe a variance estimator and prove it has a small mean squared loss. An important case to consider is how to apply our result when the measurements are taken at different time values, i.e., $x_{i}$ and $y_{i}$ are not observed at the same time. This is usually the case in the applications, where the measurements are done using a single channel (\cite{dodson},\cite{ludwig}). Hence, in Section \ref{sec:time}, we discuss the conditions sufficient to ignore the difference in the time values of observations, allowing us to apply our estimator as if the observations occured at the same time values. We also describe a possible way to interpolate the data such that it fits our model. Finally, in Section \ref{sec:num_results}, we present numerical results based on real and simulated data. Since in the simulation we know the exact value of the ratio, it gives us a confirmation of our theoretical results, showing the proposed estimator behaves well. In Section \ref{sec:conclusion}, we conclude this paper, giving several idea for possible future work. Section \ref{sec:proofs} contains the proofs in detail to all propositions and lemmas stated in this work.

\vspace{7 mm}


\section{The Special Case of $\|\mu_{2}\|^2=\Theta(n)$}
\label{sec:special}


Although the condition $|D(Y)|>\alpha\mathbb{E}[D]$ considered in the previous section holds with high probability, it cannot be checked by the data itself since $\mathbb{E}[D]=\|\mu_{2}\|^2$ is unknown. In the applications, it might be useful to be able to check whether this condition holds with certainty. To do that, we have to impose new restrictions on our model. Thus, we consider a special case when $\|\mu_{2}\|^2=\Theta(n)$.   

\begin{remark}
Similar to Remark \ref{remark:condition}, one should notice that the condition $\|\mu_{2}\|^2=\Theta(n)$ holds if $f$ is bounded away from zero, both from below and above, i.e., for all $x$, $c_{1}\leq f(x)\leq c_{2}$ holds for positive constants $c_{1}$ and $c_{2}$. 
\end{remark}

Now we present the results of considering this special case. Lemma \ref{lemma:Dnew} (below) is similar to Lemma \ref{lemma:D} (Section \ref{sec:intro}), with a slightly different condition. In Lemma \ref{lemma:Dnew}, we prove that $|D(Y)|>\beta n$ holds with a high probability. Notice that now we have the lower bound on the denominator in terms of $n$, thus making this condition checkable by the data. However, the lemma holds when $\|\mu_{2}\|^2\geq\beta n$. Assuring that this condition holds requires some prior knowledge of function $f$. A sufficient condition for $\|\mu_{2}\|^2\geq\beta n$ to hold wpuld be that $f$ is bounded below by $\beta$.  

In Proposition \ref{prop:msl_new}, similar to Proposition \ref{prop:msl},  we prove that the mean squared loss of our estimator is small, conditional on the event that $|D(Y)|>\beta n$. In formal terms: 

\begin{lemma}
\label{lemma:Dnew} Assuming $\|\mu_{2}\|^2\geq \beta n$, for positive constant $\beta$, the following holds:
\begin{equation*}
\mathbb{P}\left(|D(Y)|\leq \beta n\right)=O\left(\frac{1}{n}\right).
\end{equation*}
\end{lemma}

\begin{proposition}
\label{prop:msl_new}
Assume that $\|\mu_{2}\|^2=\Theta(n)$ and further assume that $\|\mu_{2}\|^2\geq \beta n$ for some positive constant $\beta$. Then the following holds
\begin{equation*}
\mathbb{E}\left[(\hat{r}-r)^2 \left|\; |D(Y)|>\beta n\right.\right]=O\left(\frac{1}{n}\right).
\end{equation*}
\end{proposition}

The proofs of Lemma \ref{lemma:Dnew} and Proposition \ref{prop:msl_new} use the same techniques as the proofs of Lemma \ref{lemma:D} and Proposition \ref{prop:msl}, respectively. Making the exposition complete, we provide the full proofs of the above statements in Section \ref{sec:proofs}.

\vspace{7 mm}


\section{Variance Estimation}\label{sec:variance}


As one can easily notice, computing our estimator $\hat{r}$ requires knowing $\sigma_{2}$, the variance of each of the components of the vector $Y$ with values in $\mathbb{R}^{n}$. To simplfy the analysis of our estimator, we have assumed that this variance is known. However, in the real applications, it is usually the case that we do not have prior knowledge of $\sigma_{2}$; thus it is important to also have a method for estimating this variance. Essentially, our problem here is to estimate the variance of the coordinates of vector $Y$.

More specifically, given the observations $(y_{1},\ldots,y_{n})$ satisfying 
\begin{equation*}
y_{i}=f(t_{i})+w_{i}, 1\leq i\leq n,
\end{equation*}
where the function $f$ is unknown and $w_{i}$'s represent i.i.d. $\mathcal{N}(0,\sigma^2)$ random variables, we estimate $\sigma^2$, suppressing the index in denoting the variance to make the notation simpler. Denoting the differences of adjacent $y_{i}$'s as $d_{i}$'s, we get 
\begin{align*}
d_{i}=y_{2i-1}-y_{2i} &=f(t_{2i-1})+w_{2i-1}-f(t_{2i})-w_{2i} \\
&\sim\mathcal{N}\left(f(t_{2i-1})-f(t_{2i}),2\sigma^2\right),\;\;\; 1\leq i\leq \lfloor n/2\rfloor=:m.
\end{align*}
Let us denote with $\nu_{i}=f(t_{2i-1})-f(t_{2i})$ the mean of $d_{i}$, $1\leq i\leq m$. Since $w_{i}, 1\leq i\leq n$, are independent random variables, $d_{i}, 1\leq i\leq m$, are also independent. We propose the following estimator for $\sigma^2$:
\begin{equation*}
\hat{\sigma}^2=\frac{1}{2m}\sum_{i=1}^{m}d_{i}^2.
\end{equation*}
In Lemma \ref{lemma:var}, we prove that the square loss of this estimator is small, assuming function $f$ has the property that $\sum_{i=1}^{m}(f(t_{2i-1})-f(t_{2i}))^2=\sum_{i=1}^{m}\nu_{i}^2=\|\nu\|^2=O(\sqrt{n})$. In formal terms:

\begin{lemma}\label{lemma:var}
Assuming $\|\nu\|^2=O(\sqrt{n})$, we have
\begin{equation*}
\mathbb{E}\left[(\hat{\sigma}^2-\sigma^2)^2\right]=O\left(\frac{1}{n}\right).
\end{equation*}
\end{lemma}

The proof of Lemma \ref{lemma:var} is given in detail in Section \ref{sec:proofs}. 

\begin{remark} 
All functions $f$ with finite quadratic variation will satisfy the assumption of Lemma \ref{lemma:var}, $\|\nu\|^2=O(\sqrt{n})$, even more $\|\nu\|^2=O(1)$. Also, if $\left|f(t_{2i-1})-f(t_{2i})\right|^2 \leq 1/\sqrt{n}$, for all $1\leq i\leq m$, the condition is satisfied.
\end{remark}

\vspace{7 mm}


\section{Discussion on Different Time Values}\label{sec:time}


In our model, we assumed that the time values at which we observe functions $f$ and $rf$ are the same. However, in the applications we might not be able to make these observations at the same time, as occured in the real data case, provided in detail in Section \ref{sec:num_results}. In other words, \cite{dodson} and \cite{ludwig} state that it is usually the case that we cannot do the measurements at the same time due to a single measuring channel. Hence, it is important to see under what specific conditions it would be possible to ignore difference in the observation times, while still getting good theoretical guarantees on the proposed estimator, which we apply as previously.

Specifically, in the case when the observations, denoted as $x_{i}$ and $y_{i}$, $1\leq i\leq n$, are made at different times, the model consists of the following:
\begin{equation*}
x_{i}=rf(p_{i})+z_{i}\;\;\;(i=1,2,\ldots,n)
\end{equation*}
and
\begin{equation*}
\label{eq:set2}
y_{i}=f(t_{i})+w_{i}\;\;\;(i=1,2,\ldots,n), 
\end{equation*}
where $t_{i}$ and $p_{i}$ are not necessarily equal. As in our previous model presented in Section \ref{sec:intro}, we have $z_{i}\sim\mathcal{N}(0,\sigma_{1}^2)$ and $w_{i}\sim\mathcal{N}(0,\sigma_{2}^2)$ representing the independent noise variables. 
We denote $X=(rf(t_{1})+z_{1},\ldots,rf(t_{n})+z_{n})$ and $X'=(x_{1},\ldots,x_n)$. Since we do not observe vector $X$, we cannot evaluate our estimator for $r$, $\hat{r}=\frac{\langle X,Y\rangle}{\|Y\|^2-n\sigma_{2}^2}$. However, in this case we use the following estimator $\hat{r}'=\frac{\langle X',Y\rangle}{\|Y\|^2-n\sigma_{2}^2}$. In Lemma \ref{lemma:time}, we prove that assuming the difference $\|X-X'\|^2=r^2\sum_{i=1}^{n}(f(t_{i})-f(p_{i}))^2$ is bounded by a constant, the expected squared loss for estimating $r$ with $\hat{r}'$ is small, conditional on the event that the denominator is large enough. In formal terms:

\begin{lemma}\label{lemma:time}
Assuming $\|\mu_{2}\|^2\in\Omega(n)$ and $\|X-X'\|=O(1)$, we get
\begin{equation*}
E\left[|r-\hat{r}'|^2|\; |D(Y)|>\alpha E[D(Y)]\right]= O\left(\frac{1}{n}\right).
\end{equation*}
\end{lemma}

In the proof of Lemma \ref{lemma:time}, we first prove that the mean squared loss in estimating $\hat{r}$ with $\hat{r}'$ is small, i.e., the expected value of $(\hat{r}-\hat{r}')^2$ conditional on the event that the denominator $D(Y)$ is $O(1/n)$. Second, we use the result of Proposition \ref{prop:msl}, which shows that the mean squared loss in estimating the true ratio $r$ with $\hat{r}$ is small. Combining these two through a triangle inequality proves Lemma \ref{lemma:time}. A detail proof is also given in Section \ref{sec:proofs}.

\vspace{7 mm}


\section{Numerical Results}\label{sec:num_results}


This estimation problem was identified in the work of researchers in the Department of Earth, Atmospheric and Planetary Sciences at MIT. Here, we present our result applied to the data they collected (Figure \ref{fig:data}). The observations can be written as
$x_{i}=rf(t_{i})+z_{i}$ and $y_{i}=f(p_{i})+w_{i}$, $1\leq i\leq n$, where $z_{i}\sim\mathcal{N}(0,\sigma_{1}^{2})$ and $w_{i}\sim\mathcal{N}(0,\sigma_{2}^2)$ are independent, normal random variables. Since the measurements could not be taken at the same time, the sets $\{t_{i}:1\leq i\leq n\}$ and $\{p_{i}:1\leq i\leq n\}$, representing time values, are disjointed. However, the difference between the time values is small; thus we will consider these two sets to be the same. Since $\sigma_{2}$ is unknown, we also estimate $\sigma_{2}$ as described in Section \ref{sec:variance}. Figure \ref{fig:ratio} shows how the value of our estimator for the ratio depends on the number of data points used.

In this particular data set, the differences between the time values are small, thus we considered them equal for the purpose of applying our estimator. With some other data set, we might need to interpolate between the observed points to get two sequences with the same time values. The interpolation can be done in various ways; however, we need to make sure the new (interpolated) points are independent in order to apply the estimator. To make sure the new sequence consists of independent points, in the interpolation process we should not use an observed point more than once.

\begin{figure}
\centering
\begin{subfigure}{.5\textwidth}
  \centering
  \includegraphics[width=\linewidth]{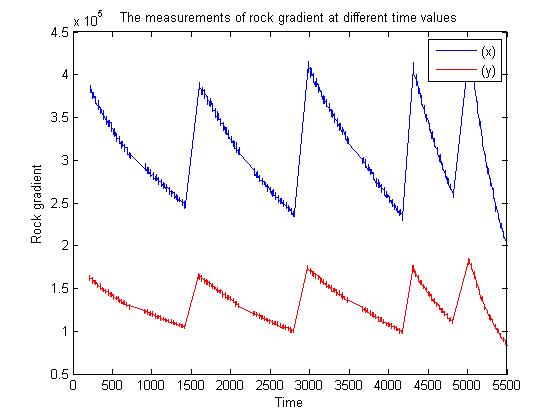}
  \caption{Data sets}
  \label{fig:data}
\end{subfigure}%
\begin{subfigure}{.5\textwidth}
  \centering
  \includegraphics[width=\linewidth]{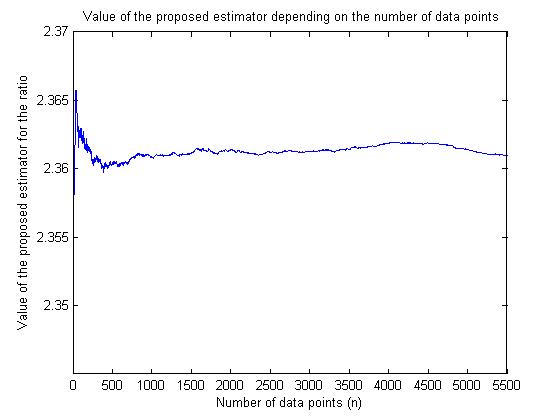}
  \caption{The dependence of estimated ratio on the number of data points used}
  \label{fig:ratio}
\end{subfigure}
\caption{Real data example}
\end{figure}

We also provide a result of a numerical simulation for which we take $x_{i}=10*(\sin(i)+2)+z_{i}$ and $y_{i}=\sin(i)+2+w_{i}$, $1\leq i\leq n$, where $z_{i}$ and $w_{i}$ are independent, identically distributed random variables from standard normal distribution. Figure \ref{fig:simulated} shows the estimated ratio $r$ using our estimator.

\begin{figure}[h]
\centering
\includegraphics[width=0.5\textwidth]{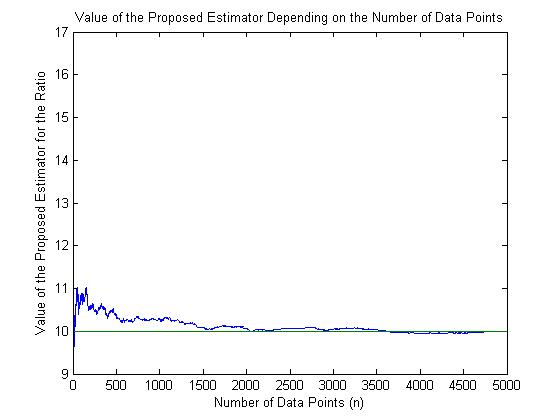}
\caption{The dependence of estimated ratio on the number of data points used (simulated data example)}
\label{fig:simulated}
\end{figure}

\vspace{7 mm}


\section{conclusion}\label{sec:conclusion}


The main significance of our work is that it not only proposes a solution for a common estimation problem that appears in various fields, e.g., chemistry, physics, geology; but our work also provides satisfying theoretical guarantees. Specifically, we proved our proposed estimator for the ratio of two functions is both consistent and has a small mean squared loss. We also solved various issues often encountered while working with real data, such as estimating noise variance. Furthermore, we proposed an approach for when data does not follow our model closely, when there are time discrepancies in observations. For this case, we provided a sufficient condition for which our proposed solution works well, i.e., with a small loss. Also, we mentioned the interpolation method that modifies the given data so that it fits our model, having observations at the same time. However, in the future work, the error coming from both interpolation and estimation procedure could be computed. Also, in this work, we have assumed the noise is normally distributed. Asuming different distribution on the noise, it would be both useful and interesting to analyze our estimator or propose a new one.

\vspace{7 mm}


\section{Proofs}\label{sec:proofs}


\begin{proof}[Proof of Proposition \ref{prop:infinite_loss}]
Le us denote $D(y)=\|y\|^2-n\sigma_{2}^2$.
Computing the conditional expectation $E_{X|Y=y}$ of the squared loss $(\hat{r}-r)^2$, we get
\begin{align}
\mathbb{E}_{X|Y=y}[(\hat{r}-r)^2] &=\mathbb{E}_{X|Y=y}\left[\left(\frac{\langle X,Y\rangle-r(\|Y\|^2-n\sigma_{2}^2)}{\|Y\|^2-n\sigma_{2}^2}\right)^2 \right]  \nonumber
\\ &= \frac{1}{D(y)^2}\left(\mathbb{E}_{X}\left[\langle X,y\rangle^2\right]-2rD(y)\mathbb{E}_{X}\left[\langle X,y\rangle\right]+r^2D(y)^2\right) \nonumber
\\ &= \frac{1}{D(y)^2}(\langle\mu_{1},y\rangle^2+\sigma_{1}^2\|y\|^2-2rD(y)\langle\mu_{1},y\rangle+r^2D(y)^2) \nonumber
\\ &=\frac{1}{D(y)^2}\left((\langle\mu_{1},y\rangle-rD(y))^2+\sigma_{1}^2\|y\|^2\right) \label{eq:condl},
\end{align}
where we used the facts that $\mathbb{E}[\langle X,y\rangle]=\langle\mu_{1},y\rangle$ and $\mathbb{E}[\langle X,y\rangle^2]=\langle\mu_{1},y\rangle^2+\sigma_{1}^2\|y\|^2$. Since $(\langle\mu_{1},y\rangle-rD(y))^2$ is nonnegative for all $y\in\mathbb{R}^{n}$, (\ref{eq:condl}) implies the following inequality:
\begin{equation*}
\mathbb{E}_{X|Y=y}[(\hat{r}-r)^2]\geq\frac{\sigma_{1}^2}{D(y)^2}\|y\|^2.
\label{eq:condl}
\end{equation*}
Using the law of iterated expectation, we get 
\begin{equation*}
\mathbb{E}\left[(\hat{r}-r)^2\right]\geq \sigma_{1}^2\mathbb{E}\left[\frac{\|Y\|^2}{D(Y)^2}\right].
\end{equation*}
Random variable $\frac{\|Y\|^2}{\sigma_{2}^2}$ is distributed as noncentral chi-squared distribution $\chi'^{2}_{n}(\lambda)$ with $n$ degrees of freedom and noncentrality parameter $\lambda=\frac{\|\mu_{2}\|^2}{\sigma_{2}^2}$. 
For $K\sim Poisson\left(\frac{\lambda}{2}\right)$ we have $\chi'^{2}_{n}(\lambda)\sim\chi^{2}_{n+2K}$. 

Let us denote with $Z(k)$ a random variable with  $\chi^{2}_{n+2k}$ distribution. Then $\frac{\|Y\|^2}{\sigma_{2}^2}$ has the same distribution as $Z(K)$.
Using this we can write
\begin{align}
\nonumber
\mathbb{E}\left[\frac{\|Y\|^2}{D(Y)^2}\right] &= \frac{1}{\sigma_{2}^2}\mathbb{E}\left[\frac{\frac{\|Y\|^2}{\sigma_{2}^{2}}}{\left(\frac{\|Y\|^2}{\sigma_{2}^2}-n\right)^2}\right]=\frac{1}{\sigma_{2}^2}\mathbb{E}\left[\frac{Z(K)}{(Z(K)-n)^2}\right] \\ 
& =\frac{1}{\sigma_{2}^2}\mathbb{E}_{K}\left[\mathbb{E}_{Z(K)|K=k}\left[\frac{Z(k)}{(Z(k)-n)^2}\;\Bigg|\; K=k\right]\right],
\label{eq:chi}
\end{align}
where in the last identity we used the law of iterated expectation. 

Let us further denote $d=n+2k$, where $k$ is fixed constant. Using the fact that $Z(k)\sim\chi_{d}^2$, we can write the inner expectation in \ref{eq:chi} in terms of the density function of chi-sqaured distribution:
\begin{align*}
\mathbb{E}\left[\frac{Z(k)}{(Z(k)-n)^2}\right] &=\int_{0}^{\infty} \frac{x}{(x-n)^2} \frac{x^{d/2-1}e^{-x/2}}{2^{d/2}\Gamma(d/2)}dx\\
&=\int_{0}^{n}\frac{x^{d/2}}{(x-n)^2} \frac{e^{-x/2}}{2^{d/2}\Gamma(d/2)}dx+\int_{n}^{\infty}\frac{x^{d/2}}{(x-n)^2} \frac{e^{-x/2}}{2^{d/2}\Gamma(d/2)}dx\,.
\end{align*}
In the first integral, all the terms except $\frac{x^{d/2}}{(x-n)^2}$ are bounded on $[0,n]$. Since $\int_{0}^{n}\frac{x^{d/2}}{(x-n)^2}dx=+\infty$, the first integral in the sum above is $+\infty$. Consequently, $E\left[\frac{Z(k)}{(Z(k)-n)^2}\right]$ is also $+\infty$, implying further that the expression in (\ref{eq:chi}) is infinite. Finally, this proves the mean squared loss is also infinite, finishing the proof.

\end{proof}


\begin{lemma}
\label{lemma:4thnorm}
$Z\sim\mathcal{N}(\mu,\sigma^2 I_n)$, $\mu\in\mathbb{R}^n$, implies 
\begin{equation*}
\mathbb{E}[\|Z\|^2]=\|\mu\|^2+n\sigma^2\;\;\text{and}\;\;Var[\|Z\|^2]=2n\sigma^4+4\sigma^2\|\mu\|^2.
\end{equation*}
\end{lemma}

\begin{proof}
Denote $\mu=(\mu_1,\ldots,\mu_n)$ and $Z=(Z_1,\ldots,Z_n)$, where $Z_i\overset{iid}{\sim}\mathcal{N}(\mu_i,\sigma)$, $1\leq i\leq n$.
Using the fact that the fourth moment of a $\mathcal{N}(\nu,\sigma^2)$ normal random variable is $\nu^4+6\sigma^2\nu^2+3\sigma^4$, we get
\begin{align*}
\mathbb{E}\left[\|Z\|^4\right] &= \mathbb{E}\left[(Z_{1}^2+\ldots+Z_{n}^2)^2\right]
     = \mathbb{E}\left[\sum_{i=1}^{n} Z_{i}^4 + 2\sum_{i<j}^{}Z_{i}^2 Z_{j}^2\right]
\\ & = \sum_{i=1}^{n}(\mu_{i}^4+6\mu_{i}^2\sigma^2+3\sigma^4)+2\sum_{i<j}^{}(\mu_{i}^2+\sigma^2)(\mu_{j}^2+\sigma^2 )
\\ & = \sum_{i=1}^{n}\mu_{i}^4+6\sigma^2\|\mu\|^2+3n\sigma^4+2\sum_{i<j}^{}\mu_{i}^2\mu_{j}^2+2(n-1)\sigma^2\|\mu\|^2+n(n-1)\sigma^4
\\ & = \|\mu\|^4+n(n+2)\sigma^4+2(n+2)\|\mu\|^2\sigma^2.
\end{align*}

The above computation implies $\text{Var}[\|Z\|^2]=\|\mu\|^4+n(n+2)\sigma^4+2(n+2)\|\mu\|^2\sigma^2-\|\mu\|^4-2n\|\mu\|^2\sigma^2-n^2\sigma^4=2n\sigma^4+4\sigma^2\|\mu\|^2$.

\end{proof}


\begin{lemma}
\label{lemma:numerator}
For $X\sim\mathcal{N}(\mu_{1},I_{n}\sigma_{1}^2)$ and $Y\sim\mathcal{N}(\mu_{2},I_{n}\sigma_{2}^2)$, the following holds
\begin{equation*}
\mathbb{E}\left[(\langle X,Y\rangle-rD)^2\right]=O\left(\|\mu_{2}\|^2+n\right),
\end{equation*}
where $D=\|Y\|^2-n\sigma_{2}^2$.
\end{lemma}

\begin{proof}
Let us write $X=\mu_{1}+\sigma_{1}\epsilon_{1}$ and $Y=\mu_{2}+\sigma_{2}\epsilon_{2}$, where $\epsilon_{1}\sim\mathcal{N}(0,I_{n})$ and $\epsilon_{2}\sim\mathcal{N}(0,I_{n})$ are independent normally distributed random variables. Using this representation of $X$ and $Y$, we compute
\begin{align*}
\langle X,Y\rangle-rD &= \langle \mu_{1}+\sigma_{1}\epsilon_{1}, \mu_{2}+\sigma_{2}\epsilon_{2}\rangle-r(\langle\mu_{2}+\sigma_{2}\epsilon_{2},\mu_{2}+\sigma_{2}\epsilon_{2}\rangle-n\sigma_{2}^2)
\\ &= \langle \mu_{1},\mu_{2}\rangle + \sigma_{2}\langle \mu_{1},\epsilon_{2} \rangle+\sigma_{1}\langle \epsilon_{1},\mu_{2}\rangle+\sigma_{1}\sigma_{1}\langle \epsilon_{1},\epsilon_{2}\rangle
\\ &\;\;\;\;-r(\|\mu_{2}\|^2+2\sigma_{2}\langle \mu_{2},\epsilon_{2}\rangle)+ \sigma_{2}^2(\|\epsilon_{2}\|^2-n)
\\ &= \sigma_{1}\langle\epsilon_{1},\mu_{2}\rangle+\sigma_{1}\sigma_{2}\langle\epsilon_{1},\epsilon_{2}\rangle-r\sigma_{2}^2(\|\epsilon_{2}\|^2-n)-r\sigma_{2}\langle\mu_{2},\epsilon_{2}\rangle.
\end{align*}

Using the following identities $\mathbb{E}\left[\langle\epsilon_{1},\mu_{2}\rangle^2\right]=\|\mu_{2}\|^2$, $\mathbb{E}\left[\langle\epsilon_{1},\epsilon_{2}\rangle^2\right]=n$, $\mathbb{E}\left[\langle\mu_{2},\epsilon_{2}\rangle^2\right]=\|\mu_{2}\|^2$, $\mathbb{E}\left[\|\epsilon_{2}\|^4\right]=3n+n(n-1)=n^2+2n$, we get 
\begin{align*}
\mathbb{E}\left[(\langle X,Y-rD\rangle)^2\right] &= \mathbb{E}\left[\left(\sigma_{1}\langle\epsilon_{1},\mu_{2}\rangle+\sigma_{1}\sigma_{2}\langle\epsilon_{1},\epsilon_{2}\rangle-r\sigma_{2}^2(\|\epsilon_{2}\|^2-n)-r\sigma_{2}\langle\mu_{2},\epsilon_{2}\rangle \right)^2\right]
\\ &= \sigma_{1}^2 \mathbb{E}\left[\langle\epsilon_{1},\mu_{2}\rangle^2\right]+\sigma_{1}^2\sigma_{2}^2 \mathbb{E}\left[\langle\epsilon_{1},\epsilon_{2}\rangle^2 \right]+r^2\sigma_{2}^4 \mathbb{E}\left[(\|\epsilon_{2}\|^2-n)^2\right]
\\ &\;\;\;\;+r^2 \sigma_{2}^2 \mathbb{E}\left[\langle\mu_{2},\epsilon_{2})^2\right]+2r^2\sigma_{2}^3\mathbb{E}\left[(\|\epsilon_{2}\|^2-n)\langle\mu_{2},\epsilon_{2}\rangle\right] 
\\ &= \sigma_{1}^2\|\mu_{2}\|^2+n\sigma_{1}^2\sigma_{2}^2+r^2\sigma_{2}^4(n^2+2n-2n^2+n^2)+r^2\sigma_{2}^2\|\mu_{2}\|^2+0 
\\ &= \sigma_{1}^2\|\mu_{2}\|^2+n\sigma_{1}^2\sigma_{2}^2+2nr^2\sigma_{2}^4+r^2\sigma_{2}^2\|\mu_{2}\|^2
\\ &\leq \text{max}\{\sigma_{1}^2, \sigma_{1}^2\sigma_{2}^2,2r^{2}\sigma_{2}^4, r^{2}\sigma_{2}^2\}(\|\mu_{2}\|^2+n)
\\ &= O(\|\mu_{2}\|^2+n).
\end{align*}

\end{proof}


\begin{proof}[ Proof of Lemma \ref{lemma:D}]

Since $\alpha\in(0,1)$, we get
\begin{align}
\nonumber
\mathbb{P}\left(|D(Y)|\leq\alpha \mathbb{E}[D(Y)]\right)&\leq\mathbb{P}\left(\left|D(Y)-\mathbb{E}[D(Y)]\right|> (1-\alpha)\mathbb{E}[D(Y)]\right)\\
&\leq \frac{\text{Var}[D(Y)]}{(1-\alpha)^2\mathbb{E}[D(Y)]^2},
\label{eq:chebyshev}
\end{align}
where in the second inequality we used Chebyshev's inequality. Using Lemma \ref{lemma:4thnorm}, we get $\text{Var}[D(Y)]=\text{Var}[\|Y\|^2]=2n\sigma_{2}^4+4\sigma_{2}^2\|\mu_{2}\|^2=O(\|\mu_{2}\|^2+n)$. Hence, from \ref{eq:chebyshev} we conclude
\begin{equation*}
\mathbb{P}\left(|D(Y)|\leq\alpha\mathbb{E}[D(Y)]\right)=O\left(\frac{\|\mu_{2}\|^2+n}{\|\mu_{2}\|^4}\right).
\end{equation*}

\end{proof}


\begin{proof}[Proof of Proposition \ref{prop:msl}]

To keep the notation easier let us write $D$ instead of $D(Y)$. Denote $E^{*}=\mathbb{E}\left[(\hat{r}-r)^2 \left\vert\; |D|>\alpha \mathbb{E}[D]\right.\right]$ and $P^{*}=\mathbb{P}\left(|D|>\alpha \mathbb{E}[D]\right)$. 

Using 
\begin{equation*}
E^{*} \leq \frac{1}{\alpha^2\mathbb{E}[D]^2}\mathbb{E}\left[(\langle X,Y\rangle-rD)^2 \;\Big|\;|D|>\alpha \mathbb{E}[D]\right],
\end{equation*}
we get 
\begin{align}
E^{*}P^{*} &\leq \frac{1}{\alpha^2 \mathbb{E}[D]^2} \mathbb{E}\left[(\langle X,Y\rangle-rD)^2 \;\Big|\; |D|>\alpha \mathbb{E}[D]\right] \mathbb{P}\left(|D|>\alpha \mathbb{E}[D]\right) \nonumber
\\ &\leq \frac{1}{\alpha^2 \mathbb{E}[D]^2} \mathbb{E}\left[(\langle X,Y\rangle-rD)^2 \right] \leq c_{1}\frac{\|\mu_{2}\|^2+n}{\|\mu_{2}\|^4}, \label{eq:product} 
\end{align}
for some constant $c_{1}$, where in the last inequality we used Lemma \ref{lemma:numerator}.

From Lemma \ref{lemma:D}, we get $P^{*}\geq 1-c_{2}\frac{\|\mu_{2}\|^2+n}{\|\mu_{2}\|^4}$, for some constant $c_{2}$. Using this in (\ref{eq:product}), we get the following upper bound on $E^{*}$ 
\begin{equation*}
E^{*}\leq \frac{c_{1}(\|\mu_{2}\|^2+n)}{\|\mu_{2}\|^4-c_{2}(\|\mu_{2}\|^2+n)}=O\left(\frac{\|\mu_{2}\|^2+n}{\|\mu_{2}\|^4}\right)=O\left(\frac{1}{n}\right),
\end{equation*}
where we used the assumption $\|\mu_{2}\|^2=\Omega(n)$.

\end{proof}


\begin{proof}[Proof of Proposition \ref{prop:consistency}]

Le us denote $D=\|Y\|^2-n\sigma_{2}^2$, the event $S=\{|D-\mathbb{E}[D]|\leq\mathbb{E}[D]/2\}$ and $S^{c}$ its complement.
Using the law of total probability, we can write 
\begin{align*}
\mathbb{P}\left((\hat{r}-r)^2\leq\frac{c\log(n)}{n}\right) &= \underset{P_{1}}{\underbrace{\mathbb{P}\left((\hat{r}-r)^2\leq\frac{c\log(n)}{n}\;\Big|\; S \right)}} \mathbb{P}(S) \\
&+ \underset{P_{2}}{\underbrace{\mathbb{P}\left((\hat{r}-r)^2\leq\frac{c\log(n)}{n}\;\Big|\; S^{c}\right)}}\mathbb{P}(S^{c}).
\end{align*}
It is easy to see that $P_{2}\mathbb{P}(S^{c})$ converges to zero as n tends to infinity. From the second inequality in (\ref{eq:chebyshev}) we have $\mathbb{P}(S^{c})=\mathbb{P}\left(|D-E[D]|>\frac{E[D]}{2}\right)=O\left(\frac{\|\mu_{2}\|^2+n}{\|\mu_{2}\|^4}\right)$, thus $\underset{n\rightarrow\infty}{\lim}\mathbb{P}(S^{c})=0$. Since $P_{2}$ is bounded, we get $\underset{n\rightarrow\infty}{\lim}P_{2}\mathbb{P}(S^{c})=0$.

Now, let us derive a lower bound for $P_{1}\mathbb{P}(S)$. Since the event we condition on is $S=\{D\in\left(\mathbb{E}[D]/2,3\mathbb{E}[D]/2\right)\}$, we get
\begin{align*}
P_{1} &=\mathbb{P}\left((\hat{r}-r)^2\leq\frac{c\log(n)}{n}\;\Big|\;S\right)\\
&\geq \mathbb{P}\left(\frac{4(\langle X,Y\rangle-rD)^2}{\mathbb{E}[D]^2}\leq\frac{c\log(n)}{n}\;\Big|\; S\right)\\
&= 1-\mathbb{P}\left(\left(\langle X,Y\rangle-rD\right)^2>\frac{c\log(n)}{4}\frac{\mathbb{E}[D]^2}{n}\;\Big|\; S\right).
\end{align*}
Using the Markov's inequality
\begin{equation*}
\mathbb{P}\left((\langle X,Y\rangle-rD)^2>\frac{c\log(n)}{4}\frac{\mathbb{E}[D]^2}{n}\;\Big|\; S\right)\leq 
\frac{\mathbb{E}\left[(\langle X,Y\rangle-rD)^2\;\Big|\; S \right]}{\frac{c\log(n)}{4}\frac{\mathbb{E}[D]^2}{n}},
\end{equation*}
we get
\begin{equation*}
P_{1} \geq 1-\frac{\mathbb{E}\left[(\langle X,Y\rangle-rD)^2\;\Big|\; S\right]}{\frac{c\log(n)}{4}\frac{\mathbb{E}[D]^2}{n}}.
\end{equation*}

Using the fact that 
\begin{equation*}
\mathbb{E}\left[(\langle X,Y\rangle-rD)^2\right]\geq \mathbb{E}\left[(\langle X,Y\rangle-rD)^2\;\Big|\;S\right]\mathbb{P}(S),
\end{equation*}
we can further bound $P_{1}\mathbb{P}(S)$ to get
\begin{equation*}
P_{1}\mathbb{P}(S) \geq \mathbb{P}(S)-\frac{\mathbb{E}\left[(\langle X,Y\rangle-rD)^2\right]}{\frac{c\log(n)}{4}\frac{\mathbb{E}[D]^2}{n}}.
\end{equation*}

We know that $1-\mathbb{P}(S)=\mathbb{P}(S^{c})=O\left(\frac{\|\mu_{2}\|^2+n}{\|\mu_{2}\|^4}\right)$. Lemma \ref{lemma:numerator} implies $\frac{\mathbb{E}\left[(\langle X,Y\rangle-rD)^2\right]}{\frac{c\log(n)}{4}\frac{\mathbb{E}[D]^2}{n}}=O\left(\frac{n(\|\mu_{2}\|^2+n)}{\log(n)\|\mu_{2}\|^4} \right)=O\left(\frac{n}{\log(n)\|\mu_{2}\|^2}\right)$. 
Hence, we conclude 
\begin{equation*}
\underset{n\rightarrow\infty}{\lim} P_{1}\mathbb{P}(S)=1.
\end{equation*}
This finishes the proof of consistency, i.e.,
$$\underset{n\rightarrow\infty}{\lim} \mathbb{P}\left((\hat{r}-r)^2\leq\frac{c\log(n)}{n}\right)=1.$$

\end{proof}


\begin{proof}[Proof of Lemma \ref{lemma:Dnew}]

Since $\mathbb{E}[D(Y)]=\|\mu_{2}\|^2\geq \beta n$, we get
\begin{equation}
\mathbb{P}\left(|D(Y)|\leq\beta n\right)\leq \mathbb{P}\left(\left|D(Y)\right|<\beta\frac{\mathbb{E}[D]}{\beta}\right)=O\left(\frac{1}{n}\right),
\end{equation} 
where the last equality comes from Lemma \ref{lemma:D}. One should note that $\beta/k_{1}\in(0,1)$ is necessary to apply Lemma \ref{lemma:D} here.

\end{proof}


\begin{proof}[ Proof of Proposition \ref{prop:msl_new}]

To keep the notation easier let us write $D$ instead of $D(Y)$. Denote $E^{*}=E\left[(\hat{r}-r)^2 \left\vert\; |D|>\beta n\right.\right]$ and $P^{*}=\mathbb{P}\left(|D|>\beta n\right)$. 

Using 
\begin{equation*}
E^{*}\leq\frac{1}{\beta^2n^2}\mathbb{E}\left[(\langle X,Y\rangle-rD)^2 \;\Big|\; |D|>\beta n\right],
\end{equation*}
we get 
\begin{align*}
E^{*}P^{*} &\leq \frac{1}{\beta^2 n^2} \mathbb{E}\left[(\langle X,Y\rangle-rD)^2 \;\Big|\; |D|>\beta n\right] \mathbb{P}\left(|D|>\beta n\right) \nonumber
\\ &\leq \frac{1}{\beta^2 n^2} \mathbb{E}\left[(\langle X,Y\rangle-rD)^2 \right] \leq c_{1}\frac{\|\mu_{2}\|^2+n}{n^2},
\end{align*}
for some constant $c_{1}$, where in the last inequality we used Lemma \ref{lemma:numerator}.

From Lemma \ref{lemma:Dnew} we get $P^{*}\geq 1-\frac{c_{2}}{n}$, for some constant $c_{2}$. Thus, the following holds
\begin{equation*}
E^{*}\leq \frac{c_{1}(\|\mu_{2}\|^2+n)}{n^2-c_{2}n}= O\left(\frac{\|\mu_{2}\|^2+n}{n^2}\right)=O\left(\frac{1}{n}\right),
\end{equation*}
where we used the assumption $\|\mu_{2}\|^2=\Theta(n)$.

\end{proof}


\begin{proof}[Proof of Lemma \ref{lemma:var}]

Let $C$ be a constant such that $\|\nu\|^2\leq C\sqrt{n}$ for all $n$. Then we get
\begin{align*}
\mathbb{E}\left[(\hat{\sigma}^2-\sigma^2)^2\right] &= E\left[\left(\frac{1}{2m}\sum_{i=1}^{m}d_{i}^2-\sigma^2\right)^2\right] \\
&= \mathbb{E}\left[ \left( \frac{1}{2m}\sum_{i=1}^{m}(d_{i}^2-2\sigma^2) \right)^2 \right] \\
&= \frac{1}{4m^2}\mathbb{E}\left[\sum_{i=1}^{m}(d_{i}^2-2\sigma^2)^2+2\sum_{i<j}(d_{i}^2-2\sigma^2)(d_{j}^2-2\sigma^2)\right] \\
&= \frac{1}{4m^2} \left( \sum_{i=1}^{m}\mathbb{E}\left[(d_{i}^2-2\sigma^2)^2\right]+2\sum_{i<j}\mathbb{E}\left[(d_{i}^2-2\sigma^2)(d_{j}^2-2\sigma^2)\right]\right)
\end{align*}
Using the fact that $d_{i}\sim\mathcal{N}(\nu_{i},2\sigma^2)$, we get $\mathbb{E}[d_{i}^2]=\nu_{i}^2+2\sigma^2$, $\mathbb{E}[d_{i}^4]=\nu_{i}^4+12\nu_{i}^2\sigma^2+12\sigma^4$. Also, $\mathbb{E}[(d_{i}^2-2\sigma^2)(d_{j}-2\sigma^2)]=\mathbb{E}[d_{i}^2-2\sigma^2]\mathbb{E}[d_{j}-2\sigma^2]=\nu_{i}^2\nu_{j}^2$ holds for all $i\neq j$ since $d_{i}$ and $d_{j}$ are independent for all $i\neq j$. Hence, we conclude
\begin{align*}
\mathbb{E}\left[(\hat{\sigma}^2-\sigma^2)^2\right] &= \frac{1}{4m^2}\left( \sum_{i=1}^{m}(\nu_{i}^4+8\nu_{i}^2\sigma^2+8\sigma^4)+2\sum_{i<j}\mu_{i}^2\mu_{j}^2 \right)\\
&= \frac{1}{4m^2}\left( \left(\sum_{i=1}^{m}\nu_{i}^2\right)^2+8\sigma^2\sum_{i=1}^{m}\nu_{i}^2+8m\sigma^4\right)\\
&\leq \frac{1}{4m^2}\left(C^2n+8C\sigma^2\sqrt{n}+8m\sigma^2\right),
\end{align*} 
Finally, we have that $\mathbb{E}\left[(\hat{\sigma}^2-\sigma^2)^2\right]=O(1/n)$.

\end{proof}


\begin{proof}[Proof of Lemma \ref{lemma:time}]

First, we use the Cauchy-Schwartz inequality to bound the expected square of the difference of $\hat{r}$ and $\hat{r}'$ conditional on the event that the denominator $D(Y)$ is bounded away from zero. We get
\begin{align*}
\mathbb{E}\left[|\hat{r}-\hat{r}'|^2 \;\Big|\; |D(Y)|>\alpha \mathbb{E}[D(Y)]\right] &= \mathbb{E}\left[\frac{\langle X-X',Y\rangle^2}{(\|Y\|^2-n\sigma_{2})^2}\;\Big|\; |D(Y)|>\alpha \mathbb{E}[D(Y)] \right]\\
&\leq \mathbb{E}\left[ \frac{\|X-X'\|^2\|Y\|^2}{\alpha^2 \mathbb{E}[D(Y)]^2} \;\Big|\; |D(Y)|>\alpha \mathbb{E}[D(Y)]\right].
\end{align*}

Since $\|X-X'\|^2=r^2\sum_{i=1}^{n}(f(t_{i})-f(p_{i}))^2$ (see Sections \ref{sec:time}) is deterministic and $\mathbb{E}=\|\mu_{2}\|^2$, we can write further write the above inequality as 
\begin{align*}
\mathbb{E}\left[|\hat{r}-\hat{r}'|^2 \;\Big|\; |D(Y)|>\alpha \mathbb{E}[D(Y)]\right] &\leq \frac{\|X-X'\|^2}{\alpha^2\|\mu_{2}\|^4} \mathbb{E}\left[\|Y\|^2 \;\Big|\; |D(Y)|^2>\alpha \mathbb{E}[D(Y)]\right]\\
&= \frac{\|X-X'\|^2}{\alpha^2\|\mu_{2}\|^4} (\|\mu_{2}\|^2+n\sigma_{2}^2),
\end{align*}
Since $\|X-X'\|^2=O(1)$, we get $\mathbb{E}\left[|\hat{r}-\hat{r}'|^2 \;|\; |D(Y)|>\alpha \mathbb{E}[D(Y)]\right]=O(1/\|\mu_{2}\|^2)=O(1/n)$. Using the result of Proposition \ref{prop:msl} that
\begin{equation*}
\mathbb{E}\left[|r-\hat{r}|^2 \;\Big|\; |D(Y)|>\alpha \mathbb{E}[D(Y)]\right]=O(1/n),
\end{equation*}
and the triangle inequality, $|r-\hat{r}'|^2 \leq |r-\hat{r}|^2+|\hat{r}-\hat{r}'|^2$, we get 
\begin{equation*}
\mathbb{E}\left[|r-\hat{r}'|^2|\; |D(Y)|>\alpha \mathbb{E}[D(Y)]\right]=O(1/n).
\end{equation*}

\end{proof}


\bibliographystyle{siam}
\bibliography{report}

\end{document}